\documentclass[11pt,a4paper]{article}

\usepackage[margin=1in]{geometry}
\usepackage{amsmath,amssymb,amsthm,mathtools}
\usepackage{enumitem}
\usepackage{booktabs}
\usepackage{authblk}
\usepackage{microtype}
\usepackage[colorlinks=true,linkcolor=blue,citecolor=blue,urlcolor=blue]{hyperref}
\usepackage[nameinlink,noabbrev]{cleveref}
\setlist[itemize]{leftmargin=2em,itemsep=3pt,topsep=4pt}
\setlist[enumerate]{leftmargin=2em,itemsep=3pt,topsep=4pt}
\allowdisplaybreaks

\theoremstyle{plain}
\newtheorem{theorem}{Theorem}[section]
\newtheorem{lemma}[theorem]{Lemma}
\newtheorem{corollary}[theorem]{Corollary}

\theoremstyle{definition}
\newtheorem{definition}[theorem]{Definition}

\theoremstyle{remark}

\newcommand{\Prb}{\mathbb{P}}

\begin{document}

\title{Revisiting Non-Adaptive Set Membership Schemes: Removing the Logarithmic Factor}
\author[1]{Mirza Galib Anwarul Husain Baig\thanks{Email: mirza@iitmandi.ac.in}}
\author[1]{Aman Kumar \thanks{Email: d24227@students.iitmandi.ac.in}}
\author[1]{Pawan Kumar Saw\thanks{Email: d24229@students.iitmandi.ac.in}}
\author[1]{Vikash Tripathi \thanks{Email: vikash@iitmandi.ac.in}}

\affil[1]{Indian Institute of Technology Mandi
}

\date{}
\maketitle

\begin{abstract}
We study non-adaptive bit-probe schemes for static set membership. Each query reads a fixed number of bits and answers by strict majority. We revisit the random bipartite-graph construction of Garg and Radhakrishnan \cite{GargRadhakrishnan2017}. The main observation is that a slightly stronger expansion property already bounds the number of exceptional elements. This makes a separate probabilistic estimate for exceptional elements unnecessary. A detailed union-bound calculation then gives, for every fixed $t\ge5$, a logarithm-free existence bound of $O_t(m^{1/c}n^{1-1/c})$, where $c=\lfloor(t-1)/2\rfloor$. For odd $t$, this removes the logarithmic factor from the general upper bound of Garg and Radhakrishnan \cite{GargRadhakrishnan2017}. 
\end{abstract}


\section{Introduction}
\label{sec:intro}

\subsection{Static set membership}

The static set-membership problem asks how to store a set compactly while still answering membership queries quickly. We are given a universe $[m]=\{1,2,\ldots,m\}$ and a set $S\subseteq[m]$ with $|S|\le n$. The data structure stores $S$ as a bit string. A query specifies an element $x\in[m]$ and asks whether $x\in S$.

The simplest solution is the characteristic vector of $S$. It uses $m$ bits. The bit at position $x$ is one exactly when $x\in S$. A query needs only one bit probe. This representation is wasteful when $n$ is much smaller than $m$. One can instead store the elements of $S$ directly, but then a query may need to inspect many stored elements. The bit-probe model studies the trade-off between the number of stored bits and the number of bits inspected by each query. The survey of Nicholson, Raman, and S. Srinivasa Rao \cite{NicholsonRamanRao2013} gives a broad introduction to the model and its history.

A scheme is called \emph{adaptive} if the location of a later probe can depend on the values returned by earlier probes. It is called \emph{non-adaptive} if all probe locations are fixed before any memory bit is read. Non-adaptive probes can be issued in parallel. This restriction is central to the present paper.

We write $s_N(m,n,t)$ for the minimum number of bits needed by a deterministic non-adaptive scheme that stores every set of size at most $n$ from $[m]$ and answers each query using at most $t$ probes. Distinct sets must have distinct representations, so the elementary counting bound is
\[
 s_N(m,n,t)\ge \log_2\binom{m}{n}.
\]
When $n\le m/2$, this is of order at least $n\log(m/n)$. Stronger lower bounds are known in particular parameter regimes, as discussed below.

\subsection{Four probes and a few elements}

The number of probes matters even when it is a small constant. Alon and Feige studied non-adaptive schemes with two, three, and four probes \cite{AlonFeige2009}. Alon and Feige~\cite{AlonFeige2009} proved that, for four non-adaptive
bitprobes and $\delta \leq \frac{1}{3}$, the space requirement satisfies
\[
s \geq c_3 m^{\frac{1}{1+\delta}}
n^{\frac{\delta}{1+\delta}},
\]
for some absolute constant $c_3>0$. The scheme is non explicit. 

Later work gave explicit constructions for some small-set cases. Baig and Kesh gave an explicit non-adaptive four-probe scheme for storing sets of size at most four and five \cite{BaigKeshWALCOM2018, BaigKeshTCS2020}. These results address particular small values of the set size. Our goal is different. We seek one asymptotic construction that works for every set size bound $n$ and every fixed probe count $t\ge5$.

\subsection{\texorpdfstring{General $t$-probe results}{General t-probe results}}

Garg and Radhakrishnan developed a random-graph method for non-adaptive majority schemes \cite{GargRadhakrishnan2017}. For odd $t\ge5$, their upper bound is
\[
 s_N(m,n,t)=O\!\left(
 t\,m^{\frac{2}{t-1}}n^{1-\frac{2}{t-1}}
 \log_2\frac{2m}{n}
 \right).
\]
Their proof builds a bipartite graph between universe elements and memory cells. Each element is connected to one cell in each of $t$ blocks. A query reads the cells adjacent to that element and uses majority. Two graph properties guarantee that every target set can be encoded correctly. The graph is then shown to exist by the probabilistic method.

There is also a general lower bound due to Buhrman, Miltersen, Radhakrishnan, and Venkatesh \cite{BuhrmanEtAl2002}. In one commonly studied regime, for every fixed $\varepsilon>0$ and $n\le m^{1-\varepsilon}$, with $t=o(\log m)$, their lower bound is
\[
 s_N(m,n,t)=\Omega\!\left(t\,m^{1/t}n^{1-1/t}\right).
\]
The lower bound holds even for adaptive schemes under the stated hypotheses, and therefore also applies to non-adaptive schemes. The upper and lower exponents do not match. Removing a logarithmic factor from an upper bound does not by itself settle the optimal space complexity.

\subsection{Our contribution}

We keep the graph model, the majority decoder, and the matching-based encoding method of Garg and Radhakrishnan \cite{GargRadhakrishnan2017}. We change the way the exceptional elements are handled.

In the usual proof, two graph properties are established separately. The first is expansion. The second limits the number of outside elements whose neighbourhood overlaps the target neighbourhood in a majority of its positions. The standard proof uses a logarithmic budget for these exceptional elements.

We show that a slightly stronger expansion property implies the exceptional-element bound deterministically. This means that we only need one probability estimate, for expansion. The required expansion radius remains proportional to $n$ when $t$ is fixed. The resulting block size has no logarithmic factor.

More precisely, for every fixed integer $t\ge5$, we prove
\[
 s_N(m,n,t)=O_t\!\left(m^{1/c}n^{1-1/c}\right),
 \qquad c=\left\lfloor\frac{t-1}{2}\right\rfloor.
\]
For odd $t$, this becomes
\[
 s_N(m,n,t)=O_t\!\left(
 m^{\frac{2}{t-1}}n^{1-\frac{2}{t-1}}
 \right).
\]
And for the even $t$, this becomes \[
 s_N(m,n,t)=O_t\!\left(
 m^{\frac{2}{t-2}}n^{1-\frac{2}{t-2}}
 \right).
\] The subscript on the big-$O$ indicates that the hidden constant may depend on the fixed probe count $t$.

The rest of the paper follows the order of the general upper-bound argument in Garg and Radhakrishnan \cite{GargRadhakrishnan2017}. We first define the graph model and admissibility conditions. We then show that an admissible graph gives a correct scheme. After that, we prove the existence of a suitable random graph. The last sections derive the space bound and compare it with the previous results.

\section{The graph model and admissibility}
\label{sec:framework}

\subsection{Non-adaptive graphs}

\begin{definition}[Non-adaptive $(m,s,t)$-graph]\label{def:graph}
A non-adaptive $(m,s,t)$-graph is a bipartite graph $G=(U,V;E)$ with left side $U=[m]$ and right side partitioned into $t$ disjoint blocks,
\[
 V=V_1\sqcup V_2\sqcup\cdots\sqcup V_t,
 \qquad |V_i|=s\quad(i\in[t]).
\]
Every $u\in U$ has exactly one neighbour in each block $V_i$. For $X\subseteq U$, let $\Gamma(X)$ be the set of right-side vertices adjacent to at least one element of $X$.
\end{definition}

The graph fixes the locations probed by each query. The memory consists of $ts$ bits, one for each vertex of $V$. To answer a query for $u$, the decoder reads the $t$ bits indexed by $\Gamma(u)$ and answers Yes if at least
\[
 a:=\left\lfloor\frac t2\right\rfloor+1
\]
of them equal one. This is strict-majority decoding. We also put
\[
 c:=t-a=\left\lfloor\frac{t-1}{2}\right\rfloor,
 \qquad \delta:=a-c.
\]
Thus $a+c=t$, $a>t/2$, and $\delta\in\{1,2\}$. If $t$ is odd, then $\delta=1$; if $t$ is even, then $\delta=2$. Since $t\ge5$, we have $c\ge2$.

The graph does not by itself specify the memory contents. For each target set $S$, the encoder has to assign zeros and ones so that exactly the elements of $S$ receive a Yes answer. We next state two structural conditions on the graph that guarantee the existence of such an assignment.

\subsection{The two admissibility properties}

Choose the parameters
\[
 g:=\left\lfloor\frac{cn}{\delta}\right\rfloor,
 \qquad R:=n+g+1.
\]
We first work in the range $m\ge R$. 

\begin{definition}[Admissible graph]\label{def:admissible}
A non-adaptive $(m,s,t)$-graph is called $(n,g)$-admissible if it satisfies the following two properties.
\begin{description}[leftmargin=3.2em,style=nextline]
\item[(P1) Expansion.] For every $W\subseteq[m]$ with $|W|\le R$,
\[
 |\Gamma(W)|\ge a|W|.
\]
\item[(P2) Few exceptional elements.] For every $S\subseteq[m]$ with $|S|=n$, define
\[
 T_S:=\{y\in[m]\setminus S:
 |\Gamma(y)\cap\Gamma(S)|\ge a\}.
\]
Then $|T_S|\le g$.
\end{description}
\end{definition}

An element $y\notin S$ belongs to $T_S$ if at least $a$ of its probe locations also lie in $\Gamma(S)$. If those locations are set to one while encoding $S$, the query for $y$ could receive a false Yes answer. Property (P2) says that there are only a few such elements.

The radius $R=n+g+1$ is one larger than the number $n+g$ needed in the matching argument. This extra vertex is deliberate. It is what makes the deterministic implication in the next section strict.

\subsection{Admissibility implies satisfiability}

\begin{lemma}[Admissibility implies a correct encoding]\label{lem:sat}
If $G$ is $(n,g)$-admissible, then for every $S\subseteq[m]$ with $|S|\le n$, there is an assignment of the $ts$ memory bits that answers all membership queries correctly by strict majority.
\end{lemma}

\begin{proof}
Fix a target set $S$ with $|S|\le n$. Choose a superset $S'\supseteq S$ of size exactly $n$. This is possible because $n\le m$.

By (P2), $|T_{S'}|\le g$. Define
\[
 Z:=S'\cup T_{S'}.
\]
The sets in this union are disjoint, so
\[
 |Z|=n+|T_{S'}|\le n+g=R-1.
\]
We will first give each element of $Z$ a private set of $a$ memory locations. We then set the locations for elements of $S$ to one and all other locations to zero.

Construct an auxiliary bipartite graph $H$. Its left side is $Z\times[a]$. Thus each element $u\in Z$ is copied $a$ times. The right side is $V$. Join each copy $(u,j)$ to every vertex in $\Gamma(u)$.

We verify Hall's condition. Take an arbitrary set $X\subseteq Z\times[a]$. Let $X_0$ be the set of elements of $Z$ that have at least one copy in $X$. There are at most $a$ copies of each element, so
\[
 |X|\le a|X_0|.
\]
All copies of the same element have the same neighbourhood. Therefore
\[
 N_H(X)=\Gamma(X_0).
\]
Since $X_0\subseteq Z$ and $|Z|\le R-1$, property (P1) applies and gives
\[
 |N_H(X)|=|\Gamma(X_0)|\ge a|X_0|\ge|X|.
\]
This holds for every subset $X$ of the left side. Hall's theorem therefore gives a matching that covers every vertex of $Z\times[a]$.

For each $u\in Z$, let $A_u$ be the set of the $a$ memory locations matched to its $a$ copies. Then
\[
 A_u\subseteq\Gamma(u),\qquad |A_u|=a,
\]
and the sets $A_u$, as $u$ ranges over $Z$, are pairwise disjoint.

Now assign the memory bits. Set a bit to one exactly when its location belongs to $A_u$ for some $u\in S$. Set every other bit to zero. We check the three possible kinds of query.

\smallskip
\noindent\emph{Case 1: $u\in S$.} All $a$ locations in $A_u$ contain one. Since $a>t/2$, the query answers Yes.

\smallskip
\noindent\emph{Case 2: $u\in Z\setminus S$.} The set $A_u$ is disjoint from every $A_v$ with $v\in S$. Hence all $a$ locations in $A_u$ contain zero. At most $t-a=c<a$ of the probes for $u$ can contain one, so the query answers No.

\smallskip
\noindent\emph{Case 3: $y\notin Z$.} Every one-location lies in $\Gamma(S)$ and hence in $\Gamma(S')$. Because $y\notin T_{S'}$, we have
\[
 |\Gamma(y)\cap\Gamma(S')|<a.
\]
This quantity is an integer, so it is at most $a-1$. Every one-location probed by $y$ lies in $\Gamma(S')$. Thus at most $a-1$ probes for $y$ contain one, and the query answers No. Also $y\notin S$, because $S\subseteq S'\subseteq Z$.

Every query falls into one of these three cases. The assignment is therefore correct for every element of the universe.
\end{proof}

\section{A random graph with the required expansion}
\label{sec:random}

We now prove that a graph satisfying (P1) exists. The graph is random. For each $u\in[m]$ and each block $V_i$, choose the neighbour of $u$ in $V_i$ independently and uniformly from its $s$ vertices. All choices, including choices made by different left vertices, are independent.

Put
\[
 C_t:=e^{a+1}\left(\frac at\right)^c.
\]
The constant $C_t$ depends only on $t$. Choose the block size
\[
 s:=\left\lceil\left(2C_t mR^{c-1}\right)^{1/c}\right\rceil.
\]
We will show that the probability of failure of (P1) is strictly below one.

\subsection{A simple size check}

We first check that $s\ge R$. Since $a>t/2$, we have $a/t>1/2$. Also $a\ge c+1$. It follows that
\[
 C_t=e^{a+1}\left(\frac at\right)^c
 >e^{a+1}2^{-c}
 \ge e^{c+2}2^{-c}
 =e^2\left(\frac e2\right)^c>1.
\]
The definition of $s$ gives
\[
 s^c\ge2C_t mR^{c-1}.
\]
Since $m\ge R$ and $2C_t>1$, we have $s^c>R^c$. Hence $s>R$. 

\subsection{Bounding a fixed witness}

Suppose (P1) fails. Then there is a nonempty set $W\subseteq[m]$ with $|W|\le R$ and
\[
 |\Gamma(W)|<a|W|.
\]
Write $r=|W|$. As neighbourhood sizes are integers, failure means
\[
 |\Gamma(W)|\le ar-1.
\]
Set $h=ar-1$. Thus the neighbourhood of $W$ is contained in some set $\Lambda\subseteq V$ of exactly $h$ vertices. The set $\Lambda$ can be chosen because
\[
 h<ar\le aR<tR\le ts=|V|.
\]

Fix $W$ and $\Lambda$. For each block $V_i$, let
\[
 \ell_i:=|\Lambda\cap V_i|.
\]
The blocks partition $V$, and so $\sum_{i=1}^t\ell_i=h$. Each of the $r$ elements of $W$ chooses its neighbour in block $V_i$ independently and uniformly. The probability that all $r$ neighbours in this block lie in $\Lambda$ is therefore $(\ell_i/s)^r$. Choices in different blocks are also independent. Thus
\[
 \Prb(\Gamma(W)\subseteq\Lambda)
 =\prod_{i=1}^t\left(\frac{\ell_i}{s}\right)^r.
\]
By the arithmetic-geometric mean inequality,
\[
 \prod_{i=1}^t\ell_i\le\left(\frac ht\right)^t.
\]
Consequently,
\[
 \Prb(\Gamma(W)\subseteq\Lambda)
 \le\left(\frac{h}{ts}\right)^{tr}.
\]
This estimate applies to each fixed pair $(W,\Lambda)$.

\subsection{Union bound over all witnesses}

There are $\binom mr$ choices for $W$ and $\binom{ts}{h}$ choices for $\Lambda$. The union bound gives
\[
 \Prb(\text{a failure witness of size }r)
 \le\binom mr\binom{ts}{h}
 \left(\frac{h}{ts}\right)^{tr}.
\]
We now simplify this expression. We use the standard estimate
\[
 \binom{x}{j}\le\left(\frac{ex}{j}\right)^j,
 \qquad 1\le j\le x.
\]
Applying it to both binomial coefficients yields
\begin{align*}
 \binom mr\binom{ts}{h}\left(\frac{h}{ts}\right)^{tr}
 &\le
 \left(\frac{em}{r}\right)^r
 \left(\frac{ets}{h}\right)^h
 \left(\frac{h}{ts}\right)^{tr}\\
 &=\left(\frac{em}{r}\right)^r e^h
 \left(\frac{h}{ts}\right)^{tr-h}.
\end{align*}
The exponent in the last factor is easy to compute. Since $h=ar-1$ and $t-a=c$,
\[
 tr-h=tr-ar+1=cr+1.
\]
Also, $h\le ar$. Hence
\[
 e^h\le e^{ar},\qquad
 \left(\frac{h}{ts}\right)^{cr+1}
 \le\left(\frac{ar}{ts}\right)^{cr+1}.
\]
It follows that
\begin{align*}
 \Prb(\text{a failure witness of size }r)
 &\le\left(\frac{em}{r}\right)^r e^{ar}
 \left(\frac{ar}{ts}\right)^{cr+1}\\
 &=\frac{ar}{ts}
 \left[\frac{em}{r}e^a
 \left(\frac{ar}{ts}\right)^c\right]^r\\
 &=\frac{ar}{ts}
 \left[e^{a+1}\left(\frac at\right)^c
 \frac{mr^{c-1}}{s^c}\right]^r\\
 &=\frac{ar}{ts}
 \left(\frac{C_tmr^{c-1}}{s^c}\right)^r.
\end{align*}

Every failure of (P1) has a witness of some size $r$ with $1\le r\le R$. A second union bound over these sizes therefore gives
\begin{equation}\label{eq:P1-sum}
 \Prb(\textup{(P1) fails})
 \le\sum_{r=1}^{R}\frac{ar}{ts}
 \left(\frac{C_tmr^{c-1}}{s^c}\right)^r.
\end{equation}

\subsection{Summing the probability bound}

The choice of $s$ gives
\[
 s^c\ge2C_t mR^{c-1}.
\]
For every $r\le R$, we therefore have
\[
 \frac{C_tmr^{c-1}}{s^c}
 \le\frac12\left(\frac rR\right)^{c-1}
 \le\frac12.
\]
We used $c\ge2$ here. Substituting this into \eqref{eq:P1-sum}, we obtain
\[
 \Prb(\textup{(P1) fails})
 \le\frac{a}{ts}\sum_{r=1}^{R}r2^{-r}.
\]
The standard geometric-series identity gives
\[
 \sum_{r=1}^{\infty}r2^{-r}=2.
\]
Thus
\[
 \Prb(\textup{(P1) fails})\le\frac{2a}{ts}.
\]
We already showed that $s>R$. Also $g\ge1$, because $c\ge2$, $n\ge1$, and $\delta\le2$. It follows that $R=n+g+1\ge3$ and $s\ge3$. Since $a<t$,
\[
 \Prb(\textup{(P1) fails})
 \le\frac{2a}{ts}<\frac2s\le\frac23<1.
\]
Therefore the random graph satisfies (P1) with positive probability.

\subsection{Property (P2) follows from expansion}

We have proved that the random graph satisfies (P1) with positive probability. We now show that every graph satisfying (P1) automatically satisfies (P2). This is the step that removes the need for a second probability estimate.

\begin{lemma}[Expansion implies the exceptional-set property]\label{lem:P1-implies-P2}
With $g$ and $R$ chosen as above, every graph satisfying (P1) also satisfies (P2).
\end{lemma}

\begin{proof}
Suppose that a graph satisfies (P1). Fix an arbitrary set $S\subseteq[m]$ with $|S|=n$. We will show that $|T_S|\le g$.

Assume, for a contradiction, that $|T_S|\ge g+1$. Choose a set $T\subseteq T_S$ with exactly $g+1$ elements. Every element of $T_S$ lies outside $S$, so $S$ and $T$ are disjoint. Define
\[
 W:=S\cup T.
\]
Then
\[
 |W|=|S|+|T|=n+g+1=R.
\]
Property (P1) applies to $W$, and therefore
\[
 |\Gamma(W)|\ge aR=a(n+g+1).
\]

We now count the same neighbourhood from above. Each element of $S$ has exactly $t$ neighbours, so $|\Gamma(S)|\le tn$. For every $y\in T$, the definition of $T_S$ gives
\[
 |\Gamma(y)\cap\Gamma(S)|\ge a.
\]
The vertex $y$ has exactly $t$ neighbours. At least $a$ of them already lie in $\Gamma(S)$. Hence at most $t-a=c$ of them lie outside $\Gamma(S)$, and
\[
 |\Gamma(y)\setminus\Gamma(S)|\le c.
\]
Every vertex of $\Gamma(W)$ either belongs to $\Gamma(S)$ or is a neighbour outside $\Gamma(S)$ of some element of $T$. Thus
\begin{align*}
 |\Gamma(W)|
 &\le |\Gamma(S)|+
       \sum_{y\in T}|\Gamma(y)\setminus\Gamma(S)|\\
 &\le tn+c|T|\\
 &=tn+c(g+1).
\end{align*}

The lower and upper bounds cannot both hold. Indeed,
\begin{align*}
 a(n+g+1)-\bigl(tn+c(g+1)\bigr)
 &=(a-t)n+(a-c)(g+1)\\
 &=-cn+\delta(g+1).
\end{align*}
By the definition $g=\lfloor cn/\delta\rfloor$,
\[
 g\le\frac{cn}{\delta}<g+1.
\]
Since $\delta>0$, this implies $\delta(g+1)>cn$. The difference above is therefore strictly positive, so
\[
 a(n+g+1)>tn+c(g+1).
\]
This contradicts the two bounds on $|\Gamma(W)|$.

We conclude that $|T_S|\le g$. Since $S$ was arbitrary, (P2) holds for every set of size $n$.
\end{proof}

\section{The admissible graph exists}
\label{sec:existence}

We can now prove the main existence theorem. The proof follows the random-graph method of Garg and  \cite{GargRadhakrishnan2017}, but the second graph property is obtained by the deterministic lemma rather than a separate probability estimate.

\begin{theorem}[Existence of an admissible graph]\label{thm:existence}
Let $m,n,t$ be integers with $1\le n\le m$ and $t\ge5$. Define
\[
 a=\left\lfloor\frac t2\right\rfloor+1,
 \quad c=t-a,
 \quad \delta=a-c,
 \quad g=\left\lfloor\frac{cn}{\delta}\right\rfloor,
 \quad R=n+g+1.
\]
Assume $m\ge R$, and set
\[
 C_t=e^{a+1}\left(\frac at\right)^c,
 \qquad
 s=\left\lceil\left(2C_t mR^{c-1}\right)^{1/c}\right\rceil.
\]
Then there is a non-adaptive $(m,s,t)$-graph satisfying (P1) and (P2).
\end{theorem}

\begin{proof}
Generate the random graph described in \Cref{sec:random}. The preceding calculation shows that the probability of failure of (P1) is at most $2a/(ts)$, which is strictly less than one. Hence the probability that (P1) holds is positive.

By \Cref{lem:P1-implies-P2}, every graph satisfying (P1) also satisfies (P2). Thus the event that both properties hold contains the event that (P1) holds. In particular,
\[
 \Prb(\textup{(P1) and (P2) hold})
 \ge1-\frac{2a}{ts}>0.
\]
An event with positive probability is nonempty. Therefore at least one graph satisfies both properties.
\end{proof}

\section{The resulting space bound}
\label{sec:space}

\begin{theorem}[Logarithm-free space bound]\label{thm:space}
For every fixed integer $t\ge5$, there is a deterministic non-adaptive strict-majority scheme for set membership using
\[
 s_N(m,n,t)=O_t\left(m^{1/c}n^{1-1/c}\right)
 \quad\text{bits},
 \qquad c=\left\lfloor\frac{t-1}{2}\right\rfloor,
\]
for all integers $1\le n\le m$.
\end{theorem}

\begin{proof}
Assuming $m\ge R$. By \Cref{thm:existence}, an admissible graph exists with
\[
 s=\left\lceil\left(2C_t mR^{c-1}\right)^{1/c}\right\rceil
\]
vertices in each block. By \Cref{lem:sat}, this graph gives a correct scheme with $ts$ memory bits.

The definition of $R$ gives
\[
 R=n+\left\lfloor\frac{cn}{\delta}\right\rfloor+1
 \le n+\frac{cn}{\delta}+1
 \le\left(2+\frac c\delta\right)n,
\]
where the last inequality uses $1\le n$. Therefore
\begin{align*}
 s&\le\left(2C_t mR^{c-1}\right)^{1/c}+1\\
 &\le\left(2C_t\left(2+\frac c\delta\right)^{c-1}\right)^{1/c}
 m^{1/c}n^{(c-1)/c}+1.
\end{align*}
The factor in front depends only on $t$. Also, since $m\ge n$,
\[
 m^{1/c}n^{(c-1)/c}
 =n\left(\frac mn\right)^{1/c}\ge n\ge1.
\]
Thus the additive one can be absorbed into the same constant. After multiplying by the fixed number $t$ of blocks, the total space is
\[
 O_t\left(m^{1/c}n^{1-1/c}\right).
\]

\end{proof}

\begin{corollary}[Odd and even probe counts]\label{cor:odd-even}
For every fixed odd integer $t\ge5$,
\[
 s_N(m,n,t)=O_t\left(
 m^{\frac{2}{t-1}}n^{1-\frac{2}{t-1}}
 \right).
\]
For every fixed even integer $t\ge6$,
\[
 s_N(m,n,t)=O_t\left(
 m^{\frac{2}{t-2}}n^{1-\frac{2}{t-2}}
 \right).
\]
\end{corollary}

\begin{proof}
If $t$ is odd, then $c=(t-1)/2$, so $1/c=2/(t-1)$. Substituting this into \Cref{thm:space} gives the first formula.

If $t$ is even, then $c=(t-2)/2$, so $1/c=2/(t-2)$. The second formula follows in the same way.
\end{proof}

\subsection{The five-probe case}

For $t=5$, we have $a=3$, $c=2$, and $\delta=1$. The parameters simplify to
\[
 g=2n,\qquad R=3n+1,
 \qquad C_5=e^4\left(\frac35\right)^2=\frac{9e^4}{25}.
\]
Thus the block size in the theorem is
\[
 s=\left\lceil\sqrt{\frac{18e^4}{25}\,m(3n+1)}\right\rceil.
\]
Since $3n+1\le4n$ for every $n\ge1$,
\[
 s\le\frac{6\sqrt{2}\,e^2}{5}\sqrt{mn}+1
 <13\sqrt{mn}+1.
\]
Hence the total space is $O(\sqrt{mn})$. A simple explicit statement is that the scheme uses at most $5\lceil13\sqrt{mn}\rceil$ bits in the range $m\ge3n+1$. When $m\le3n$, the replicated characteristic vector already uses $5m=O(\sqrt{mn})$ bits. Thus the same asymptotic bound holds for all $1\le n\le m$.

\begin{corollary}[Five probes]
There exists a non-adaptive strict-majority five-probe scheme that stores every subset of size at most $n$ from $[m]$ using $O(\sqrt{mn})$ bits.
\end{corollary}

\section{Comparison and discussion}
\label{sec:discussion}

\subsection{Comparison with the Garg--Radhakrishnan  \cite{GargRadhakrishnan2017} bound}

For odd $t\ge5$, Garg and Radhakrishnan prove the upper bound
\[
 O\!\left(
 t\,m^{\frac{2}{t-1}}n^{1-\frac{2}{t-1}}
 \log_2\frac{2m}{n}
 \right)
\]
using a random graph and majority decoding \cite{GargRadhakrishnan2017}. Our result gives, for each fixed odd $t\ge5$,
\[
 s_N(m,n,t)=O_t\!\left(
 m^{\frac{2}{t-1}}n^{1-\frac{2}{t-1}}
 \right).
\] And for the even $t$,  \[
 s_N(m,n,t)=O_t\!\left(
 m^{\frac{2}{t-2}}n^{1-\frac{2}{t-2}}
 \right).\]
The logarithmic factor is absent.

The reason is not a stronger concentration inequality. It is a change in the logical order of the graph-property proof. The original admissibility framework proves expansion and the exceptional-set bound separately. Here, expansion through the radius $R=n+g+1$ forces the exceptional-set bound automatically. Since $g$ is proportional to $n$ for fixed $t$, the radius $R$ is also proportional to $n$. The expansion union bound then leads to a block size of order $m^{1/c}n^{1-1/c}$.

\subsection{Relation to lower bounds}

The information-theoretic lower bound is $\log_2\binom{m}{n}$. The stronger lower bound of Buhrman et al. applies even to adaptive schemes in its stated regime and has the form
\[
 \Omega\!\left(t\,m^{1/t}n^{1-1/t}\right)
\]
for $n\le m^{1-\varepsilon}$ and $t=o(\log m)$, with fixed $\varepsilon>0$ \cite{BuhrmanEtAl2002}. Since non-adaptive schemes are a subclass of adaptive schemes, that lower bound also applies to our model in the same regime.

For odd $t\ge5$, the exponent of $m$ in our upper bound is $2/(t-1)$, while the lower bound has exponent $1/t$. These do not match. Our result therefore does not determine the optimal non-adaptive space complexity. It removes one logarithmic factor from a known upper bound.

\subsection{Nonconstructive existence}

Our argument uses the probabilistic method. It proves that a suitable graph exists, but it does not give a deterministic algorithm for finding that graph. Once a graph is fixed, the query pattern is simple: each element has one fixed location in each block, and the answer is the strict majority of the $t$ read bits. The encoding proof reduces to a bipartite matching problem for the particular target set.

The constants are explicit in the proof, but they have not been optimized. The hidden constant in $O_t(\cdot)$ depends on the fixed value of $t$. We make no claim of a uniform constant when $t$ grows with $m$.

\section{Conclusion}
\label{sec:conclusion}

We gave a detailed random-graph argument for logarithm-free upper bounds in the non-adaptive bit-probe model. The main observation is simple. A graph that expands every set up to a carefully chosen radius cannot have too many outside elements whose neighbourhoods overlap a target neighbourhood in a majority of positions. This turns the exceptional-set condition into a consequence of expansion.

The remaining probability calculation is a union bound over small sets with unusually small neighbourhoods. It gives a graph with the required expansion at block size $O_t(m^{1/c}n^{1-1/c})$. Hall's theorem then supplies the private memory locations needed by the encoder. For odd $t\ge5$, the result removes the logarithmic factor from the upper bound of Garg and Radhakrishnan \cite{GargRadhakrishnan2017}. The construction is nonconstructive, and the known lower bounds leave a gap. Closing that gap, or finding explicit graphs with the same parameters, remains open within the scope of this paper.

\section{Declaration of Generative AI and AI-Assisted Technologies}

During the preparation of this manuscript, the author used ChatGPT and Claude to improve the clarity and presentation of the writing, review mathematical proofs, and identify possible errors. The author carefully reviewed and verified the content and takes full responsibility for the final manuscript.


\begin{thebibliography}{20}

\bibitem{AlonFeige2009}
N.~Alon and U.~Feige,
\newblock On the power of two, three and four probes,
\newblock in \emph{Proceedings of the Twentieth Annual ACM--SIAM Symposium on Discrete Algorithms (SODA 2009)}, pp.~346--354, 2009.

\bibitem{BaigKeshWALCOM2018}
M.~G.~A.~H.~Baig and D.~Kesh,
\newblock Two new schemes in the bitprobe model,
\newblock in \emph{WALCOM: Algorithms and Computation, 12th International Conference}, Lecture Notes in Computer Science, vol.~10755, pp.~68--79, Springer, 2018.
\newblock \href{https://doi.org/10.1007/978-3-319-75172-6_7}{doi:10.1007/978-3-319-75172-6\_7}.

\bibitem{BaigKeshTCS2020}
M.~G.~A.~H.~Baig and D.~Kesh,
\newblock Two improved schemes in the bitprobe model,
\newblock \emph{Theoretical Computer Science}, 806:543--552, 2020.
\newblock \href{https://doi.org/10.1016/j.tcs.2019.08.033}{doi:10.1016/j.tcs.2019.08.033}.

\bibitem{BuhrmanEtAl2002}
H.~Buhrman, P.~B.~Miltersen, J.~Radhakrishnan, and S.~Venkatesh,
\newblock Are bitvectors optimal?,
\newblock \emph{SIAM Journal on Computing}, 31(6):1723--1744, 2002.
\newblock \href{https://doi.org/10.1137/S0097539702405292}{doi:10.1137/S0097539702405292}.

\bibitem{GargRadhakrishnan2017}
M.~Garg and J.~Radhakrishnan,
\newblock Set membership with non-adaptive bit probes,
\newblock in \emph{34th Symposium on Theoretical Aspects of Computer Science (STACS 2017)}, Leibniz International Proceedings in Informatics, vol.~66, Article 38, pp.~38:1--38:13, 2017.
\newblock \href{https://doi.org/10.4230/LIPIcs.STACS.2017.38}{doi:10.4230/LIPIcs.STACS.2017.38}.

\bibitem{NicholsonRamanRao2013}
P.~K.~Nicholson, V.~Raman, and S.~S.~Rao,
\newblock A survey of data structures in the bitprobe model,
\newblock in \emph{Space-Efficient Data Structures, Streams, and Algorithms}, Lecture Notes in Computer Science, vol.~8066, pp.~303--318, Springer, 2013.
\newblock \href{https://doi.org/10.1007/978-3-642-40273-9_19}{doi:10.1007/978-3-642-40273-9\_19}.

\end{thebibliography}
\end{document}